\documentclass[a4paper,UKenglish,cleveref, autoref, thm-restate]{lipics-v2021}

\title{A Faster Algorithm for Maximum Flow in Directed Planar Graphs with Vertex Capacities} 

\titlerunning{A Faster Algorithm for Max Flow in Directed Planar Graphs with Vertex Capacities} 

\author{Julian Enoch}{Department of Computer Science, University of Texas at Dallas, United States of
America}{}{}{Research partially supported by NSF grant CCF-1942597.}

\author{Kyle Fox}{Department of Computer Science, University of Texas at Dallas, United States of
America}{}{}{Research partially supported by NSF grant CCF-1942597.}

\author{Dor Mesica}{Efi Arazi School of Computer Science, The Interdisciplinary Center Herzliya, Israel}{}{}{}

\author{Shay Mozes}{Efi Arazi School of Computer Science, The Interdisciplinary Center Herzliya, Israel}{}{}{}

\authorrunning{J. Enoch, K. Fox, D. Mesica and S. Mozes} 

\Copyright{Julian Enoch, Kyle Fox, Dor Mesica and Shay Mozes} 



\category{} 

\relatedversion{} 

\nolinenumbers 

\hideLIPIcs  

\EventEditors{John Q. Open and Joan R. Access}
\EventNoEds{2}
\EventLongTitle{42nd Conference on Very Important Topics (CVIT 2016)}
\EventShortTitle{CVIT 2016}
\EventAcronym{CVIT}
\EventYear{2016}
\EventDate{December 24--27, 2016}
\EventLocation{Little Whinging, United Kingdom}
\EventLogo{}
\SeriesVolume{42}
\ArticleNo{23}

\usepackage{graphicx}
\usepackage{amssymb}
\usepackage{amsthm}
\usepackage{amsmath}
\usepackage{hyperref}
\usepackage{algorithm}
\usepackage{algpseudocode}
\usepackage{todonotes}
\usepackage[toc,page]{appendix}

\newcommand{\rev}{\text{rev}}
\newcommand{\ex}{\mathrm{ex}}
\newcommand{\vio}{\mathrm{vio}}

\newcommand{\inout}{(x^{in},x^{out})}
\newcommand{\outin}{(x^{out},x^{in})}
\newcommand{\gx}{g^\times}
\newcommand{\fx}{f^\times}
\newcommand{\Gx}{G^\times}
\newcommand{\xin}{x^{in}}
\newcommand{\xo}{x^{out}}
\newcommand{\Kx}{K^\times}
\newcommand{\Vx}{V^\times}
\newcommand{\push}{\textsf{Push}}
\newcommand{\relabel}{\textsf{Relabel}}
\newcommand{\pr}{\textsf{Push-Relabel}}

\newcommand{\fp}{\rho}
\newcommand{\hm}{h_{max}}
\newcommand{\Hm}{H_{max}}
\newcommand{\bp}{\textsf{Bulk-Push}}
\newcommand{\Tbp}{T_{BP}}

\begin{document}

\maketitle

\begin{abstract}
We give an $O(k^3 n \log n \min(k,\log^2 n)  \log^2(nC))$-time 
algorithm for computing maximum integer flows in
planar graphs with integer arc {\em and vertex} capacities bounded by $C$, and $k$ sources and sinks. This improves by a factor of 
$\max(k^2,k\log^2 n)$ 
over the fastest algorithm previously known for this problem [Wang, SODA 2019].

The speedup is obtained by two independent ideas. First we replace an iterative procedure of Wang that uses $O(k)$ invocations of an $O(k^3 n \log^3 n)$-time algorithm for maximum flow algorithm in a planar graph with $k$ apices [Borradaile et al., FOCS 2012, SICOMP 2017], by an alternative procedure that only makes one invocation of the algorithm of Borradaile et al. 
Second, we show two alternatives for computing flows in the $k$-apex graphs that arise in our
modification of Wang's procedure faster than the algorithm of Borradaile et al.  
In doing so, we introduce and analyze a sequential implementation of the parallel highest-distance
push-relabel algorithm of Goldberg and Tarjan~[JACM 1988]. 
\end{abstract}
\section{Introduction}
The maximum flow problem has been extensively studied in many different settings and variations.
This work concerns two related variants of the maximum flow problem in planar graphs. 
The first variant is the problem of computing a maximum flow in a directed planar network with integer arc {\em and vertex} capacities, and $k$ sources and sinks.
The second variant, which is used in algorithms for the first variant, is the problem of computing a maximum flow in a directed network that is nearly planar; there is a set of $k$ vertices, called apices, whose removal turns the graph planar.

The problem of maximum flow in a planar graph with vertex capacities has been studied in several works since the 1990s~\cite{KhullerN94,ZhangLC08,KaplanN11,DBLP:journals/siamcomp/BorradaileKMNW17,DBLP:conf/soda/Wang19}. For a more detailed survey of the history of this problem and other relevant results see~\cite{DBLP:conf/soda/Wang19} and references therein.
Vertex capacities pose a challenge in planar graphs because the standard reduction from a flow network with vertex capacities to a flow network with only arc capacities does not preserve planarity.
The problem can be solved by algorithms for maximum flow in sparse graphs (i.e., graphs with $n$ vertices and $O(n)$ edges that are not necessarily planar).
The fastest such algorithms currently known are an $O(n^2 /\log n)$-time algorithm~\cite{Orlin13} for sparse graphs, and an $O(n^{4/3+o(1)}C^{1/3})$-time algorithm for sparse graphs with integer capacities bounded by $C$~\cite{KathuriaLS20}.
Until recently, there was no planarity exploiting algorithm for the case of more than a single source and a single sink.
Significant progress on this problem was recently made by Wang~\cite{DBLP:conf/soda/Wang19}.
Wang developed an $O(k^5n \log^3 n \log^2 (nC))$-time algorithm, where $k$ is the number of sources and sinks, and $C$ is the largest capacity of a single vertex.
This is faster than using the two algorithms for general sparse graphs mentioned above when $k=\tilde O (n^{1/5}/\log^2 C + (nC)^{1/15})$.

Wang's algorithm uses multiple calls to an algorithm of Borradaile et al.~\cite{DBLP:journals/siamcomp/BorradaileKMNW17} for computing a maximum flow 
in a $k$-apex graph with only arc capacities.
The algorithm of Borradaile et al.~\cite{DBLP:journals/siamcomp/BorradaileKMNW17} is based on an approach originally suggested by Hochstein and Weihe~\cite{HochsteinWeihe} for a slightly more restricted problem.
In Borradaile et al.'s approach, a maximum flow in a $k$-apex graph with $n$ vertices is computed by simulating the $\pr$ algorithm of Goldberg and Tarjan~\cite{DBLP:journals/jacm/GoldbergT88} on a complete graph with $k$ vertices, corresponding to the $k$ apices of the input graph.
Whenever the $\pr$ algorithm pushes flow on an arc of the complete graph, the push operation is
simulated by sending flow between the two corresponding apices in the input $k$-apex graph. This can
be done efficiently using an $O(n\log ^3 n)$ time multiple-source multiple-sink (MSMS) maximum flow algorithm in planar graphs, which is the main result of the paper of Borradaile et al.~\cite{DBLP:journals/siamcomp/BorradaileKMNW17}.
Overall, their algorithm for maximum flow in $k$-apex graphs takes $O(k^3 n \log^3 n)$ time.
Flow in $k$-apex graphs can also be computed using the algorithms for sparse graphs mentioned above.
The $O(k^3 n \log^3 n)$-time algorithm of Borradaile et al. is faster than these algorithms when $k=\tilde O(n^{1/3}/\log^2 C+(nC)^{1/9})$.

\subsection{Our results and techniques}
We improve the running time of Wang's algorithm  to $O(k^3 n \log n \min(k,\log^2 n) \log^2 (nC))$.
This is faster than Wang's result by a factor of $\max(k^2,k\log^2 n)$, 
extending the range of values of $k$ for which the planarity exploiting algorithm is the fastest known algorithm for the problem to $k=\tilde O(n^{1/3}/\log^2 C+(nC)^{1/9})$.
The improvement is achieved by two main ideas.
At the heart of Wang's algorithm is an iterative procedure for eliminating excess flow from vertices violating the capacity constraints. 
Each iteration consists of computing a circulation with some desired properties. 
Wang computes this circulation using $O(k)$ calls to the algorithm of Borradaile et al. for maximum flow in $k$-apex graphs.
We show how to compute this circulation using a constant number of invocations of the algorithm for $k$-apex graphs.
This idea alone improves on Wang's algorithm by a factor of $k$.

To further improve the running time, we modify the algorithm of Borradaile et al. for maximum flow in $k$-apex graphs~\cite{DBLP:journals/siamcomp/BorradaileKMNW17}.
The algorithm of Borradaile et al. uses the $\pr$ algorithm of Goldberg and Tarjan~\cite{DBLP:journals/jacm/GoldbergT88}. 
We introduce a sequential implementation of the parallel highest-distance $\pr$ algorithm. 
In this algorithm, which we call batch-highest-distance, 
a single operation, $\bp$, pushes flow on multiple arcs simultaneously, 
instead of just on a single arc as in  Goldberg and Tarjan's $\push$ operation. 
More specifically, we simultaneously push flow on all admissible arcs whose tails have maximum height (see Section~\ref{sec:apex-alg}).
This is reminiscent of parallel and distributed $\pr$  algorithms~\cite{DBLP:journals/jacm/GoldbergT88,DBLP:journals/siamcomp/CheriyanM89}, but our algorithm is sequential, not parallel.
We prove that the total number of $\bp$ operations performed by the batch-highest-distance algorithm is $O(k^2)$ (this should be compared to $O(k^3)$ $\push$ operations for the FIFO or highest-distance $\pr$ algorithms).
We then show that, in the case of the $k$-apex graphs that show up in Wang's algorithm,
we can implement each $\bp$ operation using a single invocation 
of the $O(n \log^3 n)$-time MSMS maximum flow algorithm for planar graphs~\cite{DBLP:journals/siamcomp/BorradaileKMNW17}. 
Hence, we can find a maximum flow in such $k$-apex graphs in $O(k^2 n \log^3 n)$ time, which is faster by a factor of $k$ than the time required by the algorithm of Borradaile et al.

We also give another way to modify the algorithm of Borradaile et al. for maximum flow in $k$-apex
graphs; the second way is better when $k = o(\log^2 n)$.  
We observe that the structure of the $k$-apex graphs that arise in Wang's algorithm allows us to
implement each of the $O(k^3)$ $\push$ operations of the FIFO $\pr$ algorithm used by Borradaile et
al. in just $O(n \log n)$ time. 
This is done using a procedure due to Miller and Naor~\cite{DBLP:journals/siamcomp/MillerN95} for the case when all sources and sinks lie on a single face.

\noindent {\bf Roadmap.} 
In Section~\ref{sec:prel} we provide preliminary background and notations.
Section~\ref{sec:apex-alg} describes the sequential implementation of the parallel highest-distance $\pr$ algorithm, and its use in an algorithm for finding maximum flow in $k$-apex graphs. In Section~\ref{sec:main} we describe how to use this $\pr$ variant to obtain an improved algorithm for computing maximum flow in planar graphs with vertex capacities.

\section{Preliminaries} \label{sec:prel}

All the graphs we consider in this paper are directed.
For a graph $G$ we use $V(G)$ and $E(G)$ to denote the vertex set and arc set of $G$, respectively.
For any vertex $v \in V(G)$, let $\deg(v)$ denote the degree of $v$ in $G$.

For a path $P$ we denote by $P[u,v]$ the subpath of $P$ that starts at $u$ and ends at $v$.
We denote by $P\circ Q$ the concatenation of two paths $P,Q$ such that the first vertex of $Q$ is the last vertex of $P$.

A \emph{flow network} is a directed graph $G$ with a capacity function $c: V(G)\cup E(G) \rightarrow [0,\infty)$ on the vertices and arcs of $G$, along with two disjoint sets $S,T \subset V(G)$ called {\em sources} and {\em sinks}, respectively.
We assume without loss of generality that sources and sinks have infinite capacities, and that, for any arc $e=(u,v) \in E(G)$, the \emph{reverse} arc $(v,u)$, denoted $\rev(e)$ is also in $E(G)$, and has capacity $c(\rev(e)) = 0$.

Let $\rho: E(G) \rightarrow [0,\infty)$.
To avoid clutter we write $\rho(u,v)$ instead on $\rho((u,v))$.
For each vertex $v$ let
$
\rho^{in}(v) = \sum_{(u,v)\in E(G)}\rho(u,v)
$, 
 and
$
\rho^{out}(v) = \sum_{(v,u)\in E(G)} \rho(v,u)
$. 
The function $\rho$ is called a {\em preflow} if it satisfies the following conservation constraint: for all $v \in V(G)\setminus(S)$,  $\rho^{in}(v) \geq \rho^{out}(v)$.
The {\em excess of a vertex $v$ with respect to a preflow} $\rho$ is defined by  
$\ex(\rho,v) =  \rho^{in}(v) - \rho^{out}(v)$.
A preflow is \emph{feasible on an arc} $e \in E(G)$ if $\rho(e) \leq c(e)$. 
It is \emph{feasible on a vertex} $v \in V(G)$ if $\rho^{in}(v) \leq c(v)$. 
A preflow is said to be \emph{feasible} if, in addition to the conservation constraint, it is feasible on all arcs and vertices.
The value of a preflow $\rho$ is defined as
$
|\rho|=\sum_{s\in S} \rho^{out}(s)-\rho^{in}(s)$. 
A preflow $f$ satisfying $\ex(f,v)=0$ for all $v\in V(G)\setminus(S \cup T)$, is called a {\em flow}.
A flow whose value is $0$ is called a \emph{circulation}.
A {\em maximum flow} is a feasible flow whose value is maximum.
\begin{remark}\label{rem:st}
The problem of finding a maximum flow in a flow network with multiple sources and sinks can be reduced to the single-source, single-sink case by adding a super source $s$ and super sink $t$, and infinite-capacity arcs $(s,s_i)$ and $(t_i,t)$ for every $s_i \in S$ and $t_i \in T$. If the original network is planar then this transformation adds two apices to the graph. Throughout the paper, whenever we refer to the graph $G$, we mean the graph $G$ after this transformation, i.e., with a single source, the apex $s$, and a single sink, the apex $t$.
\end{remark}

The {\em violation of a flow $f$ at a vertex} $v$ is defined by
$\vio(f,v) = \max \{0,f^{in}(v)-c(v)\}$. 
Thus, if $f$ is a feasible flow then $\vio(f,v)=0$ for all vertices $v$.
The violation of the flow $f$ is defined to be $\vio(f)=\max_{v\in V(G)}{\vio(f,v)}$.\footnote{We define violations only with respect to flows (rather than preflows) because we will only discuss preflows in the context of flow networks without finite vertex capacities.} 

A preflow $\rho$ is \emph{acyclic} if there is no cycle $C$ such that $\rho(e)>0$ for every arc $e\in C$.
A preflow \emph{saturates} an arc $e$ if $\rho(e)=c(e)$.

The sum of two preflows $\rho$ and $\eta$ is defined as follows. For every arc $e \in E(G)$, $(\rho+\eta)(e)=\max \{0, \rho(e)+\eta(e)-\rho(\rev(e))-\eta(\rev(e))\}$. Multiplying the preflow $\rho$ by some constant $c$ to get the flow $c\rho$ is defined as $(c\rho)(e)=c\cdot \rho(e)$ for all $e\in E(G)$.

The \emph{residual capacity} of an arc $e$ with respect to a preflow $\rho$, denoted by $c_\rho(e)$, is $c(e)-\rho(e)+\rho(\rev(e))$.
The \emph{residual graph} of a flow network $G$ with respect to a preflow $\rho$ is the graph $G$ where the capacity of every arc $e \in E(G)$ is set to $c_\rho(e)$.
It is denoted by $G_\rho$.
A path of $G$ is called \emph{augmenting} or \emph{residual} (with respect to a preflow $\rho$) if it is also a path of $G_\rho$.

Suppose $G$ and $H$ are flow networks such that every arc in $G$ is also an arc in $H$.
If $f'$ is a (pre)flow in $H$ then the \emph{restriction} of $f'$ to $G$ is the (pre)flow $f$ in $G$ defined by $f(e)=f'(e)$ for all $e \in E(G)$.

\section{An algorithm for maximum flow in $k$-apex graphs}\label{sec:apex-alg}

In this section we introduce a sequential implementation of the parallel highest-distance $\pr$ algorithm of Goldberg and Tarjan~\cite{DBLP:journals/jacm/GoldbergT88}, and use it in the algorithm of Borradaile et al.~\cite{DBLP:journals/siamcomp/BorradaileKMNW17} for maximum flow in $k$-apex graphs.
We first give a high-level description of the $\pr$ algorithm.

\subsection{The $\pr$ algorithm~\cite{DBLP:journals/jacm/GoldbergT88}}

Let $H$ be a flow network (not necessarily planar) with source $s$ and sink $t$, arc capacities $c: E(H) \rightarrow \mathbb{R}$, and no finite vertex capacities.
The $\pr$ algorithm maintains a feasible \emph{preflow} function, $\fp$, on the arcs of $H$.
A vertex $u$ is called {\em active} if $\ex(\fp,u) > 0$.
The algorithm starts with a preflow that is zero on all arcs, except for the arcs leaving the source $s$, which are saturated.
Thus, all the neighbors of $s$ are initially active.
When the algorithm terminates, no vertex is active and the preflow function is guaranteed to be a maximum flow.
The algorithm also maintains a label function $h$ (also known as distance or height function) over the vertices of $H$.
The label function $h: V(H) \rightarrow \mathbb{N}$ is {\em valid} if $h(s)=|V(H)|$, $h(t)=0$ and $h(u) \leq h(v) + 1$ for every residual arc $(u,v) \in E(H_{\fp})$.

The algorithm progresses by performing two basic operations, $\push$ and $\relabel$.
A $\push$ operation applies to an arc $(u,v)$ if $(u,v)$ is residual, $\ex(\fp, u) > 0$, and $h(u)=h(v)+1$.
The operation moves excess flow from $u$ to $v$ by increasing the flow on $e$ by $\min\{\ex(\fp,u),c(e)-\fp(e)\}$.

The other basic operation, $\relabel(u)$, assigns $u$ the label $h(u)=\min\{h(v) : (u,v) \in
E(H_{\fp})\} + 1$ and applies to $u$ only if $u$ is active and $h(u)$ is not greater than the label
of any neighbor of $u$ in $H_{\fp}$. In other words, $\relabel$ applies to an active vertex $u$ only
if the excess flow in $u$ cannot be pushed out of $u$ (because $h(u)$ is not high enough).
The algorithm performs applicable $\push$ and $\relabel$ operations until no vertex is active.

To fit our purposes, we think of the algorithm as one that
only maintains explicitly the excess $\ex(\fp,v)$ and residual capacity $c_{\fp}(e)$ of each vertex $v$ and arc $e$ of $H$.
The preflow $\fp$ is implicit.
In this view, a $\push(u,v)$ operation decreases $\ex(\fp,u)$ and $c_{\fp}(u,v)$ by $\min\{\ex(\fp,u),c_{\fp}(u,v)\}$ and increases $\ex(\fp,v)$ and $c_{\fp}(v,u)$ by the same amount.

We reformulate Goldberg and Tarjan's correctness proof of the generic $\pr$ algorithm to fit this view.
\begin{lemma}[\cite{DBLP:journals/jacm/GoldbergT88}] \label{lem:GT}
Any algorithm that performs applicable $\push$ and $\relabel$ operations in any order satisfies the following properties and invariants:
\begin{bracketenumerate}
	\item $\ex(\fp,\cdot)$ and $c_{\fp}(\cdot)$ are non-negative.\footnote{This corresponds to the function $\fp$ being a feasible preflow.}
	\item The function $h$ is  a valid labeling function.
	\item For all $v \in V$, the value of $h(v)$ never decreases, and strictly increases when $\relabel(v)$ is called.
    \item $h(v) \leq 2|V(H)|-1$ for all $v \in V(H)$.
    \item Immediately after $\push(u,v)$ is performed, either $(u,v)$ is saturated or $u$ is inactive.
\end{bracketenumerate}
\end{lemma}

\begin{proof}
Properties (1) and (5) are immediate from the definition of $\push$ and the fact that excess and residual capacities only change during $\push$ operations.
Property (2) corresponds to Lemma 3.1 in~\cite{DBLP:journals/jacm/GoldbergT88}, Property (3) is proved in Lemma 3.6 in~\cite{DBLP:journals/jacm/GoldbergT88}, and Property (4) in Lemma 3.7 in~\cite{DBLP:journals/jacm/GoldbergT88}.
\end{proof}

\begin{lemma}\cite[Lemma 3.3]{DBLP:journals/jacm/GoldbergT88}\label{lem:no-st-path}
Properties (1), (2) imply that there is no augmenting path from $s$ to $t$ at any point of the algorithm.
\end{lemma}

\begin{lemma}\cite[Lemma 3.8]{DBLP:journals/jacm/GoldbergT88}\label{lem:num-relabel}
Properties (3), (4) imply that the number of $\relabel$ operations is at most $2|V(H)|-1$ per vertex and at most $2|V(H)|^2$ overall.
\end{lemma}

\begin{lemma}\cite[Lemmas 3.9, 3.10]{DBLP:journals/jacm/GoldbergT88}\label{lem:num-push}
Properties (1)-(5) imply that the number of $\push$ operations is $O(|V(H)|^2|E(H)|)$.
\end{lemma}

By Lemma~\ref{lem:num-relabel} and Lemma~\ref{lem:num-push}, the algorithm terminates.
Since upon termination no vertex is active, the implicit preflow $\fp$ is in fact a feasible flow function. By Lemma~\ref{lem:no-st-path} $\fp$ is a maximum flow  from $s$ to $t$.

Variants of the $\pr$ algorithm differ in the order in which applicable $\push$ and $\relabel$ operations are applied.
Some variants, such as FIFO, highest-distance, maximal-excess, etc., guarantee faster termination
than the \(O(|V(H)|^2|E(H)|)\) guarantee given above.

\subsection{A sequential implementation of the parallel highest-distance $\pr$ algorithm} \label{sec:bp}

We present a sequential implementation of the parallel highest-distance $\pr$ algorithm, which we call Batch-Highest-Distance.
This algorithm attempts to push flow on multiple edges simultaneously in an operation called $\bp$.
In that sense, it resembles the parallel version of the highest-distance $\pr$ algorithm.
It is important to note, however, that $\bp$ is a sequential operation and not a parallel/distributed one.

We define $\bp$, a batched version of the $\push$ operation.
$\bp(U, W)$ operates on two sets of vertices, $U$ and $W$. It is applicable under the following requirements:
\begin{romanenumerate}
    \item $\ex(u) > 0$ for all $u \in U$.
    \item There exists an integer $h$ such that $h(u) = h$ and $h(w) = h-1$ for all $u \in U$ and $w \in W$.
    \item There is a residual arc $(u,w)$ for some $u \in U$ and $w \in W$.
\end{romanenumerate}
Note that in a regular $\pr$ algorithm, conditions (i) and (ii) imply  that $\push(u,w)$ is applicable to any residual arc $(u,w)$ with $u \in U$ and $w \in W$. Condition (iii) guarantees there is at least one such arc.
$\bp$ pushes as much excess flow as possible from vertices in $U$ to vertices in $W$ so that
after $\bp$ the following property holds:
\vspace{5pt}
\begin{enumerate}[(5*)]
\item Immediately after $\bp(U,W)$ is called, 
for all $u \in U$ and $w \in W$, either $(u,w)$ is saturated or $u$ is inactive.
\end{enumerate}

We replace property (5) with the more general property (5*) in Lemma~\ref{lem:GT} and Lemma~\ref{lem:num-push}. With this modification,
Lemmas~\ref{lem:GT},~\ref{lem:no-st-path},~\ref{lem:num-relabel} and~\ref{lem:num-push} apply to our sequential implementation. The proofs from~\cite{DBLP:journals/jacm/GoldbergT88} need no change except replacing $\push$ with $\bp$.
Hence, our variant terminates correctly with a maximum flow from $s$ to $t$.

\begin{remark*}
One may think of $\bp(U,W)$ as performing in parallel all $\push$ operations on arcs whose tail is in $U$ and whose head is in $W$. However, not every maximum flow with sources $U$ and sinks $W$ can be achieved as the sum of flows pushed by multiple $\push$ operations.
For example, consider the case where $U$ consists of a single vertex $u$, with $\ex(u)=2$, $W=\{w_1,w_2\}$, and the residual capacities of $(u,w_1)$ and $(u,w_2)$ are both 2.
$\bp(U,W)$ may push one unit of excess flow from $u$ on each of $(u,w_1)$ and $(u,w_2)$, but $\push(u,w_i)$ would push 2 units of flow on $(u,w_i)$, and no flow on the other arc.
Therefore, the correctness of this variant cannot be argued just by simulating $\bp$ by multiple $\push$ operations. Instead we chose to argue correctness by stating the generalized property~(5*).
\end{remark*}

We now discuss a concrete policy for choosing which $\bp$ and $\relabel$ operations to perform in the above algorithm.
This policy is similar, but not identical, to the highest-distance $\pr$ algorithm~\cite{DBLP:journals/jacm/GoldbergT88,DBLP:journals/siamcomp/CheriyanM89}.
As long as there is an active vertex, the algorithm repeatedly executes the following two steps, which together are called a {\em pulse}.
Let $\hm$ be the maximum label of an active vertex.
That is, $\hm = \max \{ h(v) : \ex(\rho, v)>0\}$.
Let $H_{max}$ be the set of all the active vertices whose height is $\hm$.
In the first step of the pulse, the algorithm invokes $\bp(\Hm, W)$ where $W$ is the set of all vertices $w \in V$ such that $h(w)=\hm-1$.\footnote{Formally it may be that $\bp(\Hm,W)$ is not applicable because condition (iii) is not satisfied, e.g., when $W=\emptyset$.
In such cases $\bp$ does not push any flow. Condition (iii) is essential for the termination of the generic generalized algorithm, which may repeat such empty calls to $\bp$ indefinitely.
However, we prove in Lemma~\ref{lem:num-pulses} that in our specific policy there are $O(|V(H)|^2)$ pulses, regardless of the flow pushed (or not pushed) by $\bp$ in each pulse.}
In the second step of the pulse, the algorithm applies the $\relabel$ operation to all remaining active vertices in $H_{max}$ in arbitrary order.

\begin{algorithm}
\caption{Batch-Highest-Distance$(G,c)$}\label{euclid}
\begin{algorithmic}[1]
\State Initialize $h(\cdot)$, $c_{\fp}(\cdot)$ and $\ex(\cdot)$
\While{there exists an active vertex}
\State $\hm \gets \max \{ h(v) : \ex(\fp,v)>0\}$
\State $\Hm \gets \{v \in V(H) : \ex(\fp,v) > 0,\text{ } h(v)=\hm\}$
\State $W \gets \{w \in V(H) : h(w)=\hm-1 \}$
\State $\bp(\Hm, W)$
\State Relabel all active vertices in $\Hm$ in arbitrary order
\EndWhile
\end{algorithmic}
\end{algorithm}

\begin{remark*}
The crucial difference between this policy and the highest-distance $\pr$ algorithm~\cite{DBLP:journals/jacm/GoldbergT88,DBLP:journals/siamcomp/CheriyanM89} is that in the highest-distance algorithm a vertex $u$ with height $\hm$ is relabeled as soon as no more $\push$ operations can be applied to $u$.
In contrast, our variant first pushes flow from all vertices with height $\hm$ and only then relabels all of them.
\end{remark*}

We partition the pulses into two types according to whether any vertices are relabeled in the relabel step of the pulse.
A pulse in which at least one vertex is relabeled is called \emph{saturating}.
All other pulses are called \emph{non-saturating}.\footnote{This is a generalization of the notions of saturating and non-saturating $\push$ operations in~\cite{DBLP:journals/jacm/GoldbergT88}.}

By Lemma~\ref{lem:num-relabel}, the total number of $\relabel$ operations executed by the batch-highest-distance algorithm is $O(|V(H)|^2)$.
We now prove the same bound on the number of $\bp$ operations.

\begin{lemma}\label{lem:num-pulses}
The number of pulses (and hence also the number of calls to $\bp$) executed by the batch-highest-distance algorithm is $O(|V(H)|^2)$.
\end{lemma}
\begin{proof}
Note that the $\relabel$ step of a saturating pulse consists of at least one call to $\relabel$ which strictly increases the height of an active vertex $v$ whose height (before the increase) was $\hm$.
Hence, a saturating pulse strictly increases the value of $\hm$.
The fact that the height of each vertex never decreases and is bound by $2|V(H)|$ implies that (i) there are $O(|V(H)|^2)$ saturating pulses, and (ii) the total increase in $\hm$ over all saturating $\bp$ operations is $O(|V(H)|^2)$.

As for non-saturating pulses, note that since excess flow is always pushed to a vertex with lower height, the push step of a pulse does not create excess in any vertex with height greater than or equal to $\hm$, so all vertices with height greater than $\hm$ remain inactive during the pulse.
By property (5*), for every $u\in \Hm$ and $w\in W$, either $(u,w)$ is saturated, or $u$ is inactive.
Since the pulse is non-saturating, it follows that all the vertices in $\Hm$ become inactive during the pulse.
Hence, the value of $\hm$ strictly decreases during a non-saturating pulse.
Since $\hm \geq 0$, the total decrease in $\hm$ is also $O(|V(H)|^2)$, so there are $O(|V(H)|^2)$ non-saturating pulses.
\end{proof}

Note that we do not claim that implementing the $\bp$ operation by applying applicable $\push(u,w)$ operations for $u\in U$, $w\in W$ until no more such operations can be applied would result in fewer $\push$ operations than the $O(|V(H)|^2|E(H)|)$ bound of Lemma~\ref{lem:num-push} for the generic $\pr$ algorithm.
However, in Section~\ref{sec:main} we will show a situation where each call to $\bp$ can be efficiently implemented using a single invocation of a multiple-source multiple-sink algorithm in a planar graph.

\subsection{The algorithm of Borradaile et al. for $k$-apex graphs \cite{DBLP:journals/siamcomp/BorradaileKMNW17}}
The algorithm of Borradaile et al.~\cite[Section 5]{DBLP:journals/siamcomp/BorradaileKMNW17} uses the framework of Hochstein and Weihe~\cite{HochsteinWeihe}.
Let $H$ be a graph with a set $\Vx$ of $k$ apices. Denote $V_0 = V(H) \setminus \Vx$.
The goal is to compute a maximum flow in $H$ from a source $s \in V(H)$ to a sink $t \in V(H)$.
We assume that $s$ and $t$ are apices.
This is without loss of generality since treating $s$ and $t$ as apices leaves the number of apices in $O(k)$.
Let $\Kx$ be a complete graph over $\Vx$.
The algorithm computes a maximum flow $\fp$ from $s$ to $t$ in $H$ by simulating a maximum flow computation from $s$ to $t$ in $\Kx$ using the $\pr$ algorithm.
Whenever a $\push$ operation is performed on an arc $(u,v)$ of $\Kx$ it is implemented by pushing flow from $u$ to $v$ in the graph $H_{uv}$, induced by $V_0\cup \{u,v\}$ on the residual graph of $H$ with respect to the flow computed so far.
Note that, because no vertex of $V_0$ is an apex of $H$, $H_{uv}$ is a $2$-apex graph with apices $u,v$.
Borradaile et al. use this fact to compute a maximum flow from $u$ to $v$ in $H_{uv}$ as follows.
They split $u$ into multiple copies, each incident to a different vertex $w$ for which $(u,w)$ is an arc of $H_{uv}$.
A similar process is then applied to $v$.
Note that the resulting graph is planar.
A maximum flow from $u$ to $v$ in $H_{uv}$ is equivalent to a maximum flow with sources the copies of $u$ and sinks the copies of $v$ in the resulting graph.
This flow can be computed by the multiple-source multiple-sink maximum flow algorithm (the main result in~\cite{DBLP:journals/siamcomp/BorradaileKMNW17}) in $O(|V(H)|\log^3 |V(H)|)$ time.

The correctness of implementing the $\pr$ algorithm on $\Kx$ in this way was proved by Hochstein and Weihe~\cite{HochsteinWeihe} by proving essentially that the algorithm satisfies the properties in Lemma~\ref{lem:GT}.
Borradaile et al. used the FIFO policy of $\pr$, which guarantees that the number of $\push$ operations is $O(k^3)$, so the overall running time of their algorithm is $O(k^3 |V(H)|\log^3 |V(H)|)$.

\subsection{An algorithm for maximum flow in $k$-apex graphs}\label{sec:new-apex}

We use the algorithm of Borradaile et al. for maximum flow in $k$-apex graphs from the previous section,
 but use our new batch-highest-distance $\pr$ algorithm instead of the FIFO $\pr$ algorithm
 to compute the maximum flow in $\Kx$.
 Note that, in order to implement the batch-highest-distance algorithm on $\Kx$ we only need to maintain the excess $\ex(\fp,v)$ and labels $h(v)$ of each vertex $v\in \Kx$, and to be able to implement $\bp$ so that after the execution,  property (5*) is fulfilled.
 We do not define a flow function in $\Kx$ nor do we explicitly maintain residual capacities of arcs of $\Kx$. Instead, we maintain a preflow $\fp$ in $H$, and define that an arc $(u,v)$ of $\Kx$ is residual if and only if there exists a residual path from $u$ to $v$ in $H_{\fp}$ that is internally disjoint from the vertices of $\Vx$.
 Under this definition, there is no path of residual arcs in $\Kx$ starting at $s$ and ending at $t$ if and only if there is no such path in $H$.
Since $\Kx$ has $O(k)$ vertices, by Lemma~\ref{lem:num-pulses}, the algorithm performs $O(k^2)$ pulses.

We next describe how a $\bp(U,W)$ operation in $\Kx$ is implemented.
Let $A=U \cup W$.
Let $H_A$ be the graph
obtained from $H_\rho$ by deleting the vertices $\Vx \setminus A$. 
$\bp(U,W)$ in $\Kx$ is implemented by pushing a maximum flow in $H_A$ with sources the vertices $U$ and sinks the vertices $W$, with the additional restriction that the amount of flow leaving each vertex $u \in U$ is at most the excess of $u$. 
The efficiency of the procedure depends on how fast we can compute the maximum flow in $H_A$.
We denote the time to execute a single $\bp$ operation in the graph $H_A$ by $\Tbp$. Note that $\Tbp = \Omega(k)$, as it takes $\Omega(k)$ time to construct $H_A$ from $H_\rho$.

The proof of correctness is an easy adaptation of the proof of Hochstein and Weihe~\cite{HochsteinWeihe}.
We cannot use their proof without change because Hochstein and Weihe considered only $\push$ operations along a single arc of $\Kx$ rather than the $\bp$ operations which involves more than a single pair of vertices of $\Kx$.

\begin{lemma}\label{thm:apex}
Maximum flow in $k$-apex graphs can be computed in $O(k^2 \cdot \Tbp)$ time.
\end{lemma}

\begin{proof}
We first show that
the properties (1)-(4) in the statement of Lemma~\ref{lem:GT}, and the generalized property (5*) from Section~\ref{sec:bp} hold.

Property (1) holds since $\bp(U,W)$ limits the amount of flow pushed from each vertex $u\in U$ by the excess of $u$.
Properties (3) and (4) hold without change since $\relabel$ is not changed.

To show property (5*) holds, recall that an arc $(u,w)$ of $\Kx$ is residual if there exists, in the residual graph $H_{\fp}$ with respect to the current preflow $\fp$, a residual path from $u$ to $w$ that is internally disjoint from any vertex of $\Vx$.
With this definition it is immediate that property (5*) holds, since our implementation of $\bp(U,W)$ pushes a maximum flow in $H_A$ from $U$ that is limited by the excess flow in each vertex of $U$. Hence, after $\bp(U,W)$ is executed, for every $u \in U$ and $w \in W$, either there is no residual path from $u$ to $w$ in $H_{\fp}$ that is internally disjoint from $\Vx$, or $u$ is inactive.

As for property (2), since we did not change $\relabel$, $h$ remains valid after calls to $\relabel$.
It remains to show that $h$ remains a valid labeling after $\bp(U, W)$.
Consider two vertices $a,b \in \Vx$.
We will show that after $\bp(U, W)$, either the arc $(a,b)$ of $\Kx$ is saturated (i.e., is no residual path from $a$ to $b$ in $H_{\fp}$), or $h(a) \leq h(b) + 1$.
The flow pushed (in $H_A$) by the call $\bp(U,W)$ can be decomposed into a set $\mathcal P$ of flow paths, each of which starts at a vertex of $U$ and ends at a vertex of $W$.

Assume that after performing $\bp(U, W)$ there is an augmenting $a$-to-$b$ path, $Q$ in $H_{\fp}$.
If $Q$ does not intersect any path $P \in \mathcal{P}$ then $Q$ was residual also before $\bp(U, W)$ was called, so $h(a) \leq h(b) + 1$ because $h$ was a valid labeling before the call.
Otherwise, $Q$ intersects some path in $\mathcal{P}$.
Let $c,d$ be the first and last vertices of $Q$ that also belong to paths in $\mathcal P$.
Let $P,P' \in \mathcal P$ be paths such that $c \in P$ and $d \in P'$.
Let $w\in W$ be the last vertex of $P$ and let $u\in U$ be the first vertex of $P'$.
See Figure ~\ref{fig:cross} for an illustration.
Then, before $\bp(U, W)$ was called, $Q[a,c]\circ P[c,w]$ was a residual path from $a$ to $w$, and $P'[u,d]\circ Q[d,b]$ was a residual path from $u$ to $b$.
Since $h$ was a valid labeling before the call, we have
\[
h(u) \leq h(b) + 1 \text{\quad and\quad} h(a) \leq h(w) + 1.
\]
Since $h(u)=h(w) + 1$ it follows that
\[
h(a) \leq h(w)+1 = h(u) \leq h(b) + 1,
\]
showing property (2).

\begin{figure}
\begin{center}
\includegraphics[width=0.65\textwidth]{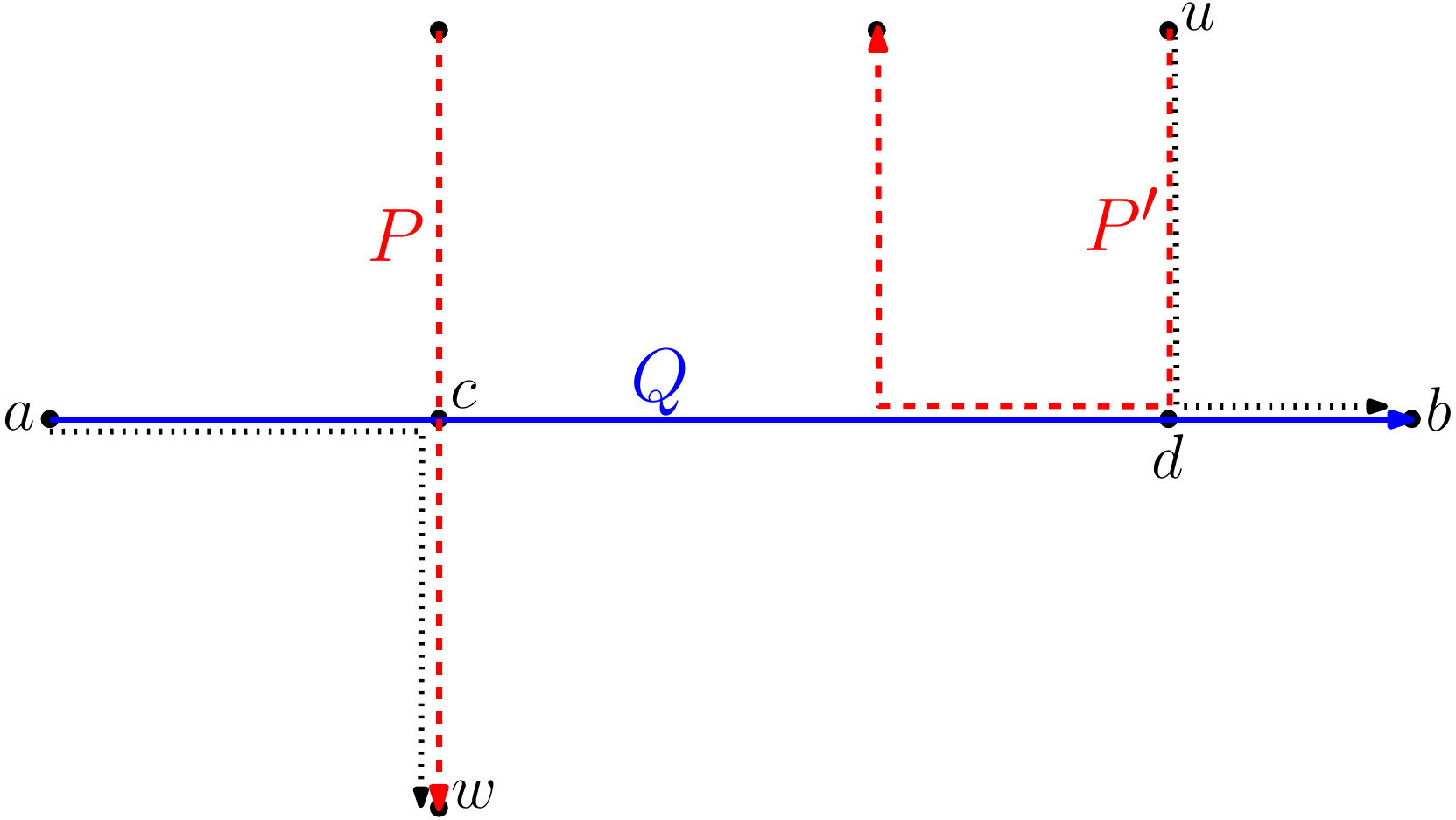}
\caption{Illustration of property (2) in the proof of Lemma~\ref{thm:apex}. A $\bp(U,W)$ operation pushes flow along paths $P$ and $P'$ (dashed red paths). If $Q$ (blue solid path) is residual after the $\bp$ operation then the dotted black paths were residual before.
\label{fig:cross}
}
\end{center}

\end{figure}

We have shown that properties (1)-(4) and (5*) hold. Hence, by Lemmas~\ref{lem:num-pulses} and~\ref{lem:num-relabel}, the algorithm terminates after performing $O(|\Vx|^2)=O(k^2)$ $\bp$ and $\relabel$ operations. Since each $\relabel$ takes $O(k)$ time, and each $\bp$ takes $\Omega(k)$ time, the total running time of the algorithm is $O(k^2 \cdot \Tbp)$.
By Lemma~\ref{lem:no-st-path}, when the algorithm terminates there is no residual path from $s$ to $t$ in $\Kx$.
By our definition of residual arcs of $\Kx$ this implies that there is no residual path from $s$ to $t$ in $H_\rho$, so $\fp$ is a maximum flow from $s$ to $t$ in $H$.
\end{proof}

\section{A faster algorithm for maximum flow with vertex capacities} \label{sec:main}

In this section, we give a faster algorithm for computing a maximum flow in a directed planar graph with integer arc and vertex capacities bounded by $C$, parameterized by the number $k$ of terminal vertices (sources and sinks).
The fastest algorithm currently known for this problem is by Wang~\cite{DBLP:conf/soda/Wang19}. It runs in
$O(k^5n\text{ polylog}(nC))$ time.
We first sketch Wang's algorithm. We only go into details in the parts of the algorithm that will be modified in our algorithm in Section~\ref{sec:gx}.

\subsection{Wang's algorithm}

Wang's algorithm uses the following two auxiliary graphs.
In both these graphs only the arcs are capacitated. Let $G$ be a planar network. Let $k$ be the total number of sources and sinks in $G$. Recall from Remark~\ref{rem:st} that we turn $G$ into a 2-apex flow network with a single super-source $s$ and super-sink $t$.
 
\begin{definition}[The graph $G^\circ$]
For a flow network $G$, the network $G^\circ$ is obtained by the following procedure. For each vertex $v \in V(G)$, replace $v$ with an undirected cycle $C_v$ with $d=\deg(v)$
vertices $v_1, ..., v_d$.\footnote{By undirected cycle we mean that there are directed arcs in both directions between every pair of consecutive vertices of the cycle $C_v$.}
Each arc in $C_v$ has capacity $c(v)/2$.
Connect each arc incident to $v$ with a different vertex $v_i$, preserving the
clockwise order of the arcs so that no new crossings are introduced.
\end{definition}

\begin{definition}[The graph $\Gx$]
Let $f$ be a flow in $G$. Let $X$ be the set of infeasible vertices,
i.e., vertices $x \in V(G)$ such that $f^{in}(x) > c(x)$.
The graph $\Gx$ is defined as follows.
Starting with $G^\circ$, for each vertex $ x \in X$, replace the cycle
representing $x$ with two vertices $x^{in}$, $x^{out}$ and an arc $\inout$ of
capacity $c(x)$.

Every arc of capacity $c > 0$ going from a vertex $u \notin C_x$ to a vertex in $C_x$ becomes an arc $(u,x^{in})$ of capacity $c$. Similarly, every arc of
capacity $c > 0$ going from a vertex of $C_x$ to a vertex $u \notin C_x$ becomes an arc $(x^{out},u)$ with capacity $c$.
\end{definition}

Note that even though $G - \{s, t\}$ is planar, $\Gx - \{s, t\}$ is not.
$\{x^{in} : x \in X\} \cup \{x^{out} : x \in X\} \cup \{s, t\}$ is an apex set in $\Gx$. Thus, $\Gx$
is a $(2|X|+2)$-apex graph.

Recall that if $H$ and $G$ are two graphs such that every arc of $G$ is also an arc of $H$, then the restriction of a flow $f'$ in $H$ to $G$ is a flow $f$ in $G$ such that $f(e)=f'(e)$ for all $e\in E(G)$.
Thus we can speak of the restriction of a flow $f^\circ$ in $G^\circ$, to a flow $f$ in $G$, and of the restriction of a flow $\fx$ in $\Gx$ to a flow $f$ in $G$.

Let $\lambda^*$ be the value of the maximum flow in $G$.
Wang's algorithm uses binary search to find $\lambda^*$.
Let $\lambda$ be the current candidate value for $\lambda^*$.
The algorithm computes a flow $f^\circ$ with value $\lambda$ in the graph $G^\circ$.
Let $f$ be the restriction of $f^\circ$ to $G$.
Wang proves that the set $X$ of infeasible vertices under $f$ has size at most $k-2$,
and that the sum of the violations of the vertices in $X$ is at most $(k-2)C$.
As long as $\vio(f)>2k$, the algorithm improves~$f$.
This improvement phase, which will be described shortly, is the crux of the algorithm.
If $\vio(f) \leq 2k$, then $O(k^2)$ iterations of the classical Ford-Fulkerson algorithm suffice to get rid of all the remaining violations.

The improvement phase of the algorithm is based on finding a circulation $g$
that cancels the violations on the infeasible vertices and does not create too much
violation on other vertices.
It can then be shown that  adding $1/k \cdot g$ to the flow
$f$ decreases $\vio(f)$ by a multiplicative factor of roughly $1-1/k$. 
After $O(k\log (kC))$
iterations of the improvement step, $\vio(f)$ is at most $2k$.

Wang proves that in order to find the circulation $g$, it suffices to compute a circulation $\gx$ in $\Gx$
that satisfies the following properties:
\begin{enumerate}
    \item $f^\times+g^\times$ is feasible in $G^\times$.
\item The restriction of $f^\times+g^\times$ to $G$ has no violations on vertices
of $X$.
\item The restriction of $f^\times+g^\times$ to $G$ has at most
$(k-2)\cdot\vio(f)$ violation on any vertex in $V(G) \setminus X$.
\end{enumerate}
The desired circulation $g$ is the restriction of $\gx$ to $G$.
If no such $\gx$ exists then $g$ does not exist, which implies that $\lambda > \lambda^*$.

Wang essentially shows that any algorithm for finding $\gx$ in $O(T)$ time, where $T=\Omega(n)$, yields an algorithm for maximum flow with vertex capacities in $O(kT\log(kC)\log(nC))$ time.
The additional terms stem from the $O(k\log (kC))$ iterations of the improvement step, and the $\log(nC)$ steps of the binary search.
Wang shows how to compute $\gx$
in $T=O(k^4n\log^3 n)$ time, by eliminating the violation at each vertex of $X$ one
after the other in an auxiliary graph obtained from $\Gx$.
Thus, the  overall running time of his algorithm is $O(k^5n\log^3n\log(kC)\log(nC))$.

\subsection{A faster algorithm for computing $\gx$}\label{sec:gx}
We propose a faster way of computing the circulation $\gx$ by eliminating the
violations in all the vertices of $X$ in a single shot.
Doing so correctly requires some care in defining the appropriate capacities in the auxiliary graph, since we only know that for each $x\in X$, $\gx$ should eliminate at least $\vio(f,x)$ units of flow from $x$, but the actual amount of flow eliminated from $x$ may have to be larger.
This issue does not come up when resolving the violations one vertex at a time as was done by Wang. 

Define an auxiliary graph $H$ as follows.
Starting with $\Gx_{\fx}$, the residual graph of $G^\times$ with respect to $f^\times$, 
\begin{itemize}
\item For each $x \in X$, set the capacity of the arc $\inout$ to be 0
and the
capacity of $\outin$ to be $c(x)$.
\item Add a super source $s'$ and arcs $(s',x^{in})$ with capacity $\vio(f,x)$ for every $x \in X$.
\item Add a super sink $t'$ and arcs $(x^{out},t')$ with capacity $\vio(f,x)$ for
every $x \in X$.
\end{itemize}
\noindent Note that $ \{s,t \} \cup \bigcup_{x\in X} \{x^{in},x^{out}\} \cup \{s',t' \}$ is an apex set of size $O(k)$ in $H$ (recall from Remark~\ref{rem:st} that $s$ and $t$ are the super source and super sink of the original graph $G$).

The circulation $\gx$ can be found using the following algorithm.
Find a maximum flow $h^\prime$ from $s'$ to $t'$ in $H$
using Lemma~\ref{thm:apex}.
Convert $h^\prime$ to an acyclic flow $h$ of the same value using the algorithm
of Sleator and Tarjan~\cite{DBLP:journals/jcss/SleatorT83} (cf.  \cite[Lemma 2.5]{DBLP:conf/soda/Wang19}).
If $h$ does not saturate every arc incident to $s'$ and $t'$, return that the desired circulation $g$ does not exist.
Otherwise,
$h$ can be extended to the desired circulation $g^\times$
by setting $g^\times\outin=h\outin + \vio(f,x)$ for every $x
\in X$ and $g^\times(e) = h(e)$ for all other arcs.

The following lemma shows that any single $\bp$ operation in the algorithm of Lemma~\ref{thm:apex} on $H$ can be implemented by a constant number of calls to the $O(n \log^3 n)$-time multiple-source multiple-sink maximum flow algorithm in planar graphs of Borradaile et al.~\cite{DBLP:journals/siamcomp/BorradaileKMNW17}. 
There are two challenges that need to be overcome.
First, the graph $H$ is $O(k)$-apex graph rather than planar.
Second, the algorithm of Borradaile et al. computes a maximum flow from multiple sources to multiple sinks, 
not a maximum flow under the restriction that each source sends at most some given limit. 
This is not a problem in the case of a single source, or a limit on just the total value of the flow, since then some of the flow pushed can be "undone". When each of the multiple sources has a different limit, undoing the flow from one source can create residual paths from another source that did not yet reach its limit.

\begin{lemma}\label{lem:time}
Any single $\bp$ operation in the execution of the algorithm of Lemma~\ref{thm:apex} on the graph $H$ defined above can be implemented in $O(n \log^3 n)$ time.
\end{lemma}
\begin{proof}
Let $\Vx = \{s,t \} \cup \bigcup_{x\in X} \{x^{in},x^{out}\} \cup \{s',t' \}$ be the set of apices of $H$.
Recall that the algorithm of Lemma~\ref{thm:apex} invokes the batch-highest-distance $\pr$ algorithm on a complete graph $\Kx$ over $\Vx$, and maintains a corresponding preflow in $H$. 
Consider a single $\bp(U, W)$ operation from a set of apices $U$ to a set of apices $W$. 
Let $\rho$ denote the preflow pushed in $H$ up to this $\bp$ operation.
Let $A = U \cup W$. 
To correctly implement $\bp(U,W)$, we find a flow $\rho'$
with sources $U$ and sinks $W$ in the graph $H$, which satisfies the following properties: 
\begin{enumerate}[(i)]
	\item For every $u \in U$, $\ex(\rho+\rho',u)\geq 0$, and 
	\item For every $u\in U$ and $w \in W$, either $\ex(\rho+\rho',u)=0$ or there is no residual path in $H_{\rho+\rho'}$ from $u$ to $w$ that is internally disjoint from $\Vx$.	
\end{enumerate}
 Condition (i) guarantees that $\rho'$ does not push more flow from a vertex $u \in U$ than the current excess of $u$. Condition (ii) is condition (5*) from Section~\ref{sec:bp}. 

Let ${H''}$ be the graph obtained from $H_\rho$ by deleting the vertices $\Vx \setminus A$. Note that the absence of residual paths that are internally disjoint from $\Vx$ in ${H''}$ is equivalent to the absence of such paths in $H$.
We will compute $\rho'$ using a constant number of invocations of the $O(n \log^3 n)$-time multiple-source multiple-sink maximum flow algorithm in planar graphs of Borradaile et al.~\cite{DBLP:journals/siamcomp/BorradaileKMNW17}. 
Instead of invoking this algorithm on ${H''}$, which is not planar, we shall invoke it on modified versions of ${H''}$ which are planar.

Starting with $H''$, we split each vertex $w \in W$ into $\deg(w)$ copies. Each arc $e$ that was incident to $w$ before the split is now incident to a distinct copy of $w$, and is embedded so that it does not cross any other arc in the graph.
Let ${H'}$ denote the resulting graph, and
let $W'$ denote the set of vertices created as a result of splitting all the vertices of $W$. 

\begin{figure}
\begin{center}
\includegraphics[width=1\textwidth]{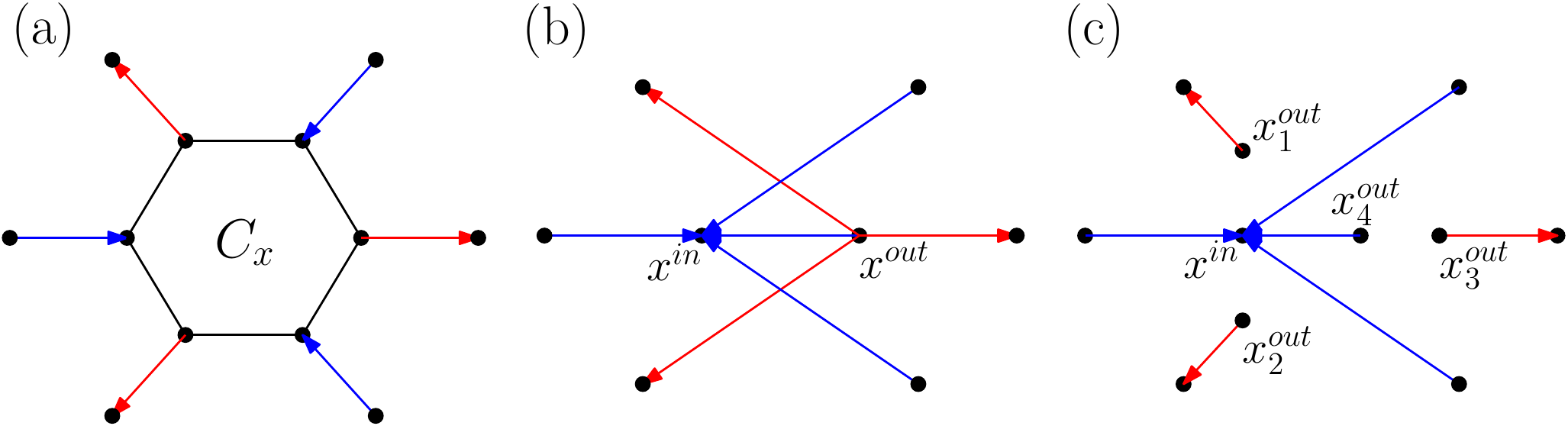}
\caption{Illustration of the auxiliary graphs used in the algorithm of Lemma~\ref{lem:time}. Only a portion of the graphs around some vertex $x \in X$ is shown. $(a)$ the graph $G^\circ$. $(b)$
the graph $H$. Note that the only crossings are between arcs incident to $\xin$ and arcs incident to $\xo$. $(c)$ the graph $H'$ in the case that $\xo$ belongs to $W$. $\xo$ is split into multiple copies, eliminating all arc crossings.}
\end{center}
\end{figure}

The set $W'$ replaces $W$ as the set of sinks of the flow $\rho'$ we need to compute. 
Note that $U$ is an apex set in ${H'}$.  
We then build the flow $\rho'$ gradually, by computing the following steps, each using a single invocation of the multiple-source multiple-sink maximum flow algorithm of Borradaile et al.
 in $O(n \log^3 n)$ time.
In what follows, when we say that the flow $\rho'$ satisfies condition (ii) for a subset $U'$ of $U$ we mean that for every $u\in U'$ and $w \in W$, either $\ex(\rho+\rho',u)=0$ or there is no residual path in $H_{\rho+\rho'}$ from $u$ to $w$ that is internally disjoint from $\Vx$.
\begin{enumerate}[(1)]
  \item \label{step:s} If $s \in U$, starting with ${H'}$, we split $s$ into $\deg(s)$ copies. 
    Each arc $e$ that was incident to $s$ before the split is now incident to a distinct copy of $s$, and is embedded so that it does not cross any other arc in the graph.
We also delete all other vertices of $U$. 
We compute in the resulting graph, which is planar, a maximum flow with sources the copies of $s$ and the sinks $W^\prime$.
Let $\rho'_s$ be the flow computed. 
If $|\rho'_s|>\ex(\rho,s)$, we decrease $|\rho'_s|$ by pushing $|\rho'_s|-\ex(\rho,s)$ units of flow back from $W'$ to the copies of $s$. 
This can be done in $O(n)$ time in reverse topological order w.r.t. $\rho'_s$ (cf.~\cite[Section 1.4]{DBLP:journals/siamcomp/BorradaileKMNW17}). 
We view $\rho'_s$ as a flow in $H$, and set $\rho' = \rho'_s$.
By construction $\rho'$ satisfies condition (i), and satisfies condition (ii) for the subset $\{s\}$.
    \item If $t \in U$, starting with $H_{\rho'}'$, we repeat step (1) with $t$ taking the role of $s$ to compute a flow $\rho'_t$.
Set $\rho' \leftarrow \rho' + \rho'_t$ 
By construction of $\rho'_t$, $\rho'$ now satisfies satisfies condition (i), and satisfies condition (ii) for the subset $U \cap \{s,t\}$.
    \item \label{step:in} Let $U^{in}$ be the set $U \cap \{\xin : x \in X\}$. 
If $U^{in} \neq \emptyset$, starting with ${{H'}}_{\rho'}$, we delete all the vertices of $U$ that are not in $U^{in}$. 
Note that, since the resulting graph does not contain $s,t,s',t'$, nor any $\xo$ for any $x \in X$, 
and since arcs incident to $\xin$ only cross those incident to $\xo$, the resulting graph is planar. 
For every $\xin \in U^{in}$ we add a vertex $x'$ and an arc $(x',\xin)$ with capacity $\ex(\rho,\xin)$. 
The resulting graph is still planar. 
We compute a maximum flow $\rho'_{in}$ with sources $\{x' : \xin \in U^{in}\}$ and sinks $W^\prime$.
We view $\rho'_{in}$ as a flow in $H$, and set $\rho' \leftarrow \rho' + \rho'_{in}$. 
By construction of $\rho'_{in}$, $\rho'$ now satisfies condition (i), and satisfies condition (ii) for the subset $U \cap (\{s,t\} \cup U^{in})$.
    \item We repeat step (\ref{step:in}) with $out$ taking the role of $in$ to compute a flow $\rho'_{out}$. 
By construction of $\rho'_{in}$, $\rho'$ now satisfies condition (i), and satisfies condition (ii) for $U \cap (\{s,t\} \cup U^{in} \cup U^{out})$.
\end{enumerate}
Since $s'$ and $t'$ are the source and sink of the flow computed by the $\pr$ algorithm, they are
never active vertices, so they never belong to $U$. Hence $\{s,t\} \cup U^{in} \cup U^{out}
\supseteq U$, and conditions (i) and (ii) are fully satisfied by $\rho'$. 
\end{proof}

Combining Lemma~\ref{lem:time} and Lemma~\ref{thm:apex}, we get the following lemma. 
\begin{lemma}\label{lem:main}
The algorithm described above finds a circulation $g^\times$ such that
\begin{enumerate}
    \item $f^\times+g^\times$ is feasible in $G^\times$.
\item The restriction of $f^\times+g^\times$ to $G$ has no violations at vertices
of $X$.
\item The restriction of $f^\times+g^\times$ to $G$ has violation at most
$(k-2)\cdot\vio(f)$ at any vertex in $V(G) \setminus X$.
\end{enumerate}
in $O(k^2 n \log^3 n)$ time if such a circulation exists.
\end{lemma}

\begin{proof}  
We first analyze the running time. 
Computing the graph $H$ can be done in $O(n)$ time.
Computing the flow $h^\prime$ in $H$ using the algorithm of Lemma~\ref{thm:apex} takes $O(k^2 \cdot \Tbp)$ time.
By Lemma~\ref{lem:time}, $\Tbp = O(n \log^3 n)$ for the graph $H$.
Transforming $h^\prime$ to an acyclic flow $h$ using the algorithm
of Sleator and Tarjan~\cite{DBLP:journals/jcss/SleatorT83} takes $O(n\log n)$ time. Finally, computing $g^\times$ from $h$ takes $O(n)$ time.
Hence, the total time to compute $\gx$ is $O(k^2 n \log^3 n)$.  

In order to prove the correctness of the algorithm, we will  first prove that 
there exists a feasible flow $h$ in $H$ that saturates every arc incident to $s'$ and $t'$ if and only if 
there exists a circulation $\gx$ in $\Gx$ that satisfies conditions (1) and (2) in the statement of the lemma.

($\Leftarrow$) Assume the circulation $\gx$ exists in $\Gx$.
Define a flow $h$ in $H$ as follows. For every arc $e \in E(H)$ not of the form $\outin$ set $h(e)=\gx(e)$.
For every $x \in X$, set $h\outin=\gx\outin-\vio(f,x)$, $h(s,x^{in})=\vio(f,x)$ and $h(x^{out}, t)=\vio(f,x)$.
Since the restriction of $\fx +\gx$ to $G$ has no violations on the vertices of $x$, $\gx \outin \geq \vio(f,x)$, so $h\outin \geq 0$ and $h$ is a well defined flow.
By definition, the flow $h$ saturates every arc incident to $s'$ and $t'$.

To show that $h$ is feasible in $H$
it is enough to show that $h\outin \leq c(x)$ for every $x \in X$ (on all other arcs $h$ is feasible because $\gx$ is feasible in $\Gx_{f^\times}$).
Let $x \in X$.
Since $\fx+\gx$ is feasible in $\Gx$, $\gx\outin \leq \fx\inout$.
Since $h\outin = \gx\outin - \vio(f,x)$, $h\outin \leq \fx\inout - \vio(f,x) = c(x)$.

($\Rightarrow$) Assume there exist a feasible flow $h$ in $H$ that saturates every arc incident to $s'$ and $t'$, and let $\gx$ be the circulation obtained from $h$ as described above.
We show that $\fx+g^\times$ is feasible in $\Gx$.
On all arcs $e$ not of the form $\outin$, $\gx(e) = h(e)$ and the capacity of $e$ in $\Gx_{\fx}$ equals the capacity of $e$ in $H$. Therefore, since $h$ is a feasible flow in $H$, $\gx$ is feasible on $e$ in $\Gx_{\fx}$, so $\fx + \gx$ is feasible on $e$ in $\Gx$.
We now focus on the arcs $(x^{out}, x^{in})$ for each
$x \in X$.
Let $e = \outin$. Observe that $0 \leq h(e) \leq c(x)$.
Since $g^\times(e)=h(e)+\vio(f,x)$ we have that $\vio(f,x) \leq g^\times(e) \leq c(x) + \vio(f,x) =
f^\times(e)$. Since $(f^\times+\gx)\inout = f^\times\inout - \gx\outin$, we have $0 \leq (f^\times+\gx)\inout \leq c(x)$, so  $f^\times + \gx$ is feasible in $\Gx$.

To finish proving the ($\Rightarrow$) direction, we show that the restriction of $f^\times + g^\times$ to $G$ has no violations
on the vertices of $X$.
By definition of $\Gx$ and of residual graph, the only arcs in $G^\times_{\fx}$ that  can carry flow out of
$x^{in}$ are the reverses of the arcs that carry flow into $x$ in $f$, and the
only arcs  that can carry flow into $x^{out}$ are the reverses of the arcs that
carry flow out of $x$ in $f$.
We will show that $(f+g)^{in}(x) \leq c(x)$ by considering separately the
contribution of the flow on arcs of $G$ that in $\Gx$ are incident to
$x^{out}$, and arcs of $G$ that in $\Gx$ are incident to $x^{in}$.

The only arc of $f^\times$ that carries flow into $x^{out}$ is
$\inout$.
Thus, there is no arc $e$ of $G$ such that $f^\times (e)$ carries flow into $x^{out}$.
Since $\gx$ only carries flow into $x^{out}$ along the reverses of
arcs that carry flow out of $x^{out}$ in $f^\times$ and since for every such arc $e'$, $\gx(e') \leq \fx(\rev(e'))$, there is also no arc $e$ of $G$ such that $(f^\times+\gx) (e)$ carries flow into $x^{out}$.

The total flow that $f^\times$ carries into $x^{in}$ is $c(x) + \vio(f,x)$. Let $z$ denote the total amount of flow that $\gx$ carries into $x^{in}$. Since the only arc incident to
$x^{in}$ that carries flow in $\gx$ and does not belong to $G$
is $\outin$, the total amount of flow that $\gx$ carries into $x^{in}$ on arcs that belong to $G$ is $z - \gx\outin$. On the other hand, $\gx$ carries $z$ units of flow out of $x^{in}$, and all of this flow is pushed along the reverses of
arcs that carry flow into $x^{in}$ in $f^\times$ (and also belong to $G$). Hence, the total amount of flow that $f^\times + \gx$ carries into $x^{in}$ on arcs that belong to $G$ is $c(x) + \vio(f,x) + (z - \gx\outin) - z$.
Therefore,
\begin{eqnarray*}
	(f+g)^{in}(x) & = &  c(x) + \vio(f,x) - \gx\outin) \\
	& = & c(x) + \vio(f,x) - (h\outin + \vio(f,x)) \\
	& = & c(x) - h\outin \\
	& \leq & c(x).
\end{eqnarray*}

We have thus shown that the algorithm computes a flow $\gx$ satisfying conditions (1) and (2) in the statement of the lemma.
To see that condition (3) is also satisfied, note that the value of the flow $h$ is $\sum_{x\in X} \vio(f,x) \leq (k-2) \cdot \vio(f)$.
Since $h$ is acyclic, $h^{in}(v) \leq (k-2)\cdot \vio(f)$ for all $v\in H$.
Since for all $v \in V(G) \setminus X$, $f^{in}(v) \leq c(v)$, and $h^{in}(v) =
(g^\times)^{in}(v)$, it follows that the violation of $f^\times + g^\times$ at $v$ is at most $(k-2) \cdot \vio(f)$.

\end{proof}

Using the $O(k^2 n\log^3 n)$-time algorithm of Lemma~\ref{lem:main} in the improvement phase of Wang's algorithm instead of using Wang's $O(k^4 n\log^3 n)$-time procedure for this phase results in a running time of $O(k^3 n\text{ polylog}(nC))$ for finding a maximum flow in $G$.

\subsection{Alternative algorithm for~\(k = o(\log^2 n)\)} \label{app:small_k}

We now provide an alternative algorithm to the one given in the previous section that is faster
for small~\(k = o(\log^2 n)\).
Specifically, we describe an algorithm for computing the circulation~\(\gx\) that runs in~\(O(k^3 n
\log n)\) time instead of the~\(O(k^2 n \log^3 n)\) time required by the algorithm of
Lemma~\ref{lem:main}.
The final running time for computing a maximum flow with integer arc and vertex capacities is
therefore \(O(k^4 n \log n \log(kC) \log(nC))\).

We use the same auxiliary graph~\(H\) as defined above and again compute a maximum flow~\(h'\)
from~\(s'\) to~\(t'\) in~\(H\).
Let $\Vx = \{s,t \} \cup \bigcup_{x\in X} \{x^{in},x^{out}\} \cup \{s',t' \} \cup S \cup T$ be the
set of apices of $H$ \emph{along with the original sources \(S\) and sinks \(T\) of \(G\)}, and let
\(\Kx\) be the complete graph on \(\Vx\).
Instead of using the batch-highest-distance \(\pr\) algorithm as in Lemma~\ref{thm:apex}, we more
directly follow the strategy of Borradaile et al.~\cite{DBLP:journals/siamcomp/BorradaileKMNW17} by
simulating a maximum flow computation from \(s\) to \(t\) in \(\Kx\) using the FIFO
\(\pr\) algorithm.
We do not wish to directly use the multiple-source multiple-sink flow algorithm of Borradaile et
al.~\cite{DBLP:journals/siamcomp/BorradaileKMNW17}, because then each of the \(O(k^3)\) \(\push\)
operations would take~\(O(n \log^3 n)\) time.
But as above, we may take advantage of the structure of~\(H\) to perform each \(\push\) operation
more quickly.

\begin{lemma}\label{lem:no_separators-time}
Any single (individual arc) $\push$ operation in the graph $H$ defined above can be implemented in
$O(n \log n)$ time.
\end{lemma}
\begin{proof}
Consider a single $\push(u, v)$ operation where \(u,v \in \Vx\). 
Let $\rho$ denote the preflow pushed in $H$ by the FIFO $\pr$ algorithm up to this
$\push$ operation.
We find a flow \(\rho'\) with source \(u\) and sink \(v\) in the graph \(H\) such that either
\(\ex(\rho + \rho', u) = 0\) or \(\ex(\rho + \rho', u) > 0\) and there is no residual path in
\(H_{\rho + \rho'}\) from \(u\) to \(v\) that is internally disjoint from \(\Vx\).

Let \({H'}\) be the graph obtained from \(H_\rho\) by deleting the vertices \(\Vx \setminus \{u,
v\}\).
Instead of invoking the \(O(n \log^3 n)\)-time multiple-source multiple-sink maximum flow algorithm
of Borradaile et al.~\cite{DBLP:journals/siamcomp/BorradaileKMNW17}, we will compute \(\rho'\) as
follows.
As before, we must consider a few different cases.
\begin{itemize}
  \item
    If \(u = s\) or \(v = s\), then \(v \in S\) or \(u \in S\), respectively.
    We push up to \(\ex(\rho, u)\) units of flow directly along the arc \((u, v)\) in constant time,
    either saturating the arc or reducing the excess flow in \(u\) to \(0\).
    We may similarly push directly along the arc \((u, v)\) in constant time if one of \(u\) or
    \(v\) is one of \(t\), \(s'\), or \(t'\) instead.
  \item
    If \(u \in \{x_1^{in}, x_1^{out}\}\) and \(v \in \{x_2^{in}, x_2^{out}\}\) for two distinct
    vertices \(x_1, x_2 \in X\), then the graph \({H'}\) is planar.
    We add a vertex \(u'\) and an arc \((u', u)\) with capacity \(\ex(\rho, u)\) and compute the
    maximum flow \(\rho'\) with source \(u'\) and sink \(v\) using the single-source single-sink
    maximum flow algorithm in planar graphs of Borradaile and
    Klein~\cite{DBLP:journals/jacm/BorradaileK09}.
  \item
    If neither of the above cases apply, then \(u,v \in \{x^{in}, x^{out}\}\) for some \(x \in X\).
    If arc \((u, v)\) has positive residual capacity, we push up to \(\ex(\rho, u)\) units of flow
    directly along it in constant time.
    Similar to Step~\ref{step:s} in the proof of Lemma~\ref{lem:time}, starting with \({H'}\), we
    split \(u\) into \(\deg(u)\) copies so that each arc that was incident to \(u\) is now
    incident to a distinct copy of \(u\).
    Similarly, we split \(v\) into \(\deg(v)\) copies so each arc that was incident to \(v\) is now
    incident to a distinct copy of \(v\).
    The resulting graph is planar, and all copies of \(u\) and \(v\) lie on a common face.
    As mentioned by Borradaile et al.~\cite[p. 1280]{DBLP:journals/siamcomp/BorradaileKMNW17}, we can then
    plug the linear time shortest paths in planar graphs algorithm of Henzinger et
    al.~\cite{DBLP:journals/jcss/HenzingerKRS97} into a divide-and-conquer procedure of Miller and
    Naor~\cite{DBLP:journals/siamcomp/MillerN95} to compute a maximum flow~\(\rho'_u\) with sources
    the copies of \(u\) and sinks the copies of \(v\) in \(O(n \log n)\) time.
    Again, if the value of this flow is greater than the excess of \(u\), we push the appropriate
    amount of flow back to the copies of \(u\) in \(O(n)\) time.
    Finally, we view \(\rho'_u\) as a flow in \(H\) to set \(\rho' = \rho'_u\).
\end{itemize}
\end{proof}

As a consequence of the previous lemma, we immediately get a variation of Lemma~\ref{lem:main} with
a running time of \(O(k^3 n \log n)\).
We use our \(O(k^3 n \log n)\) time algorithm in the improvement phase of Wang's algorithm whenever
\(k = o(\log^2 n)\). The final running time for computing a maximum flow with integer arc and vertex capacities in this case is therefore \(O(k^4 n \log n \log(kC) \log(nC))\).

\begin{theorem}
A maximum flow in an $n$-vertex planar flow network $G$ with integer arc and vertex capacities
bounded by $C$ can be computed in $O(k^3 n \log n \min(k, \log^2 n) \log(kC) \log (nC))$ time.
\end{theorem}

\end{document}